\documentclass{article}
\usepackage{graphicx} 
\usepackage{amsmath}
\usepackage{amsfonts}
\usepackage{amsthm}
\usepackage[margin=1in]{geometry}
\usepackage{xcolor}
\usepackage{appendix}

\newcommand{\rmd}{\mathrm{d}}
\newcommand{\bbE}{\mathbb{E}}\newcommand{\rme}{\mathrm{e}}

\newcommand{\bbR}{\mathbb{R}}

\newcommand{\sign}{\mathsf{sign}}

\newcommand{\erf}{\mathrm{erf}}

\newcommand{\E}{\mathbb{E}}

\newtheorem{thm}{Theorem}

\newtheorem{prop}{Proposition}
\newtheorem{cor}{Corollary}
\newtheorem{lem}{Lemma}
\newtheorem{rem}{Remark}

\title{Functional uniqueness and stability of Gaussian priors in optimal $L^1$ estimation}
\author{Leighton P. Barnes and Alex Dytso
\thanks{
Leighton P. Barnes is with Carnegie Mellon University, Pittsburgh, PA 15213, USA (e-mail: leightonb@cmu.edu). Alex Dytso is with Qualcomm Flarion Technology, Inc.,
Bridgewater, NJ 08807, USA (e-mail: odytso2@gmail.com).
}
}

\begin{document}

\maketitle

\begin{abstract}

We study when optimal Bayesian estimators under Gaussian noise are approximately linear, and what this implies about the underlying prior distribution. Consider the classical model \(Y = X + Z\), where \(Z\) is Gaussian and independent of \(X\). It is well known that under squared-error loss, the conditional mean \(\mathbb{E}[X|Y]\) is a linear function of \(Y\) if and only if the prior is Gaussian. Much less is understood under absolute-error loss, where the optimal estimator is the conditional median and standard orthogonality-based tools no longer apply.

Recent work has established that, in the Gaussian noise model, the Gaussian prior is also the unique distribution that induces an exactly linear conditional median. In this paper, we move beyond exact characterizations and develop a quantitative stability theory: if the optimal estimator is approximately linear, must the prior be close to Gaussian?

For the \(L_2\) setting, we derive explicit rates showing that near-linearity of the conditional mean forces the prior to be close to Gaussian in the Levy metric. For the \(L_1\) setting, we develop a functional-analytic framework based on Hermite expansions and adjoint operators, establishing that approximate linearity of the conditional median implies proximity to the Gaussian family.

\end{abstract}

\section{Introduction}

Bayesian estimation sits at the heart of statistics, signal processing, and information theory.  Although nonlinear estimators inevitably arise in practice, linear estimators continue to play a central role — they are easy to compute, easier to analyze, and often serve as the baseline against more sophisticated methods. Consequently, characterizing when and why optimal estimators are linear has been a long-standing  line of inquiry.

We consider the classical Gaussian observation model
\begin{equation}
Y = X + Z, \qquad Z \sim \mathcal{N}(0,1),
\end{equation}
where the noise $Z$ is taken independently of the signal $X$, and we seek to estimate $X$ from the noisy observation $Y$. The optimal estimator depends on the loss function. Under $L^2$ loss, the optimal estimator is the conditional mean; under $L^1$ loss, it is the conditional median. In the former case, a plethora of methods exist showing that Gaussian is the unique prior distribution on $X$ that induces a linear optimal estimator; the interested reader is referred to \cite{L1EstimationJournal} for the review of some of these methods. 

In contrast, the $L^{1}$ case is more subtle: the conditional median lacks orthogonality properties, has no explicit integral representation, and is not preserved by the algebraic structures that make the conditional mean tractable. Until recently, it was unclear whether Gaussians are the only priors that yield a linear conditional median. Our prior work~\cite{L1EstimationJournal} resolved this by developing an operator–theoretic framework showing that, under Gaussian noise, only Gaussian priors induce linear medians. The proof relied on properties of tempered distributions and, crucially, on the non-negativity of the prior; without it, uniqueness fails and an infinite-dimensional family of signed measures produces linearity.

This paper takes the next step in this line of work by studying a stronger property: \emph{stability}. Uniqueness tells us that only Gaussian priors lead to exactly linear estimators. Stability asks a more practical question: if the optimal estimator is \emph{almost} linear, does the prior have to be \emph{almost} Gaussian? Answering this tells us how robust the linearity phenomenon really is.

We study stability for both $L^2$ and $L^1$ losses. For the $L^2$ case, we provide an explicit quantitative rate in terms of the Lévy metric, refining and strengthening classical results. For the $L^1$ case, we introduce a Hermite expansion framework and analyze  the adjoint of the linearity-defining operator to show that the Gaussian remains the unique  solution.

\subsection{Past Work}

There is a long history of characterizing when Bayesian estimators are linear. In the Gaussian noise model with squared-error loss, the conditional mean is affine if and only if the prior is Gaussian—a fact appearing in several forms via conjugacy arguments, orthogonality principles, and analytic characterizations~\cite{diaconis1979conjugate,akyol2012conditions,kailath2000linear,L1EstimationJournal}. These ideas have motivated recent extensions to broader exponential-family and multivariate settings~\cite{dytso_meta_derivative,barnes_multivariate_gaussian_ISIT2024}.

Much less was known for absolute-error loss: the conditional median lacks an orthogonality principle and similar analytic tools. Only recently was it shown that, in the Gaussian setting, Gaussians are also the unique priors yielding linear conditional medians~\cite{L1EstimationJournal}, a result requiring analytic techniques involving convolution identities, growth bounds, and tempered distributions.

Outside the Gaussian case, the behavior diverges. Under Poisson noise, gamma priors induce affine conditional means but not linear medians~\cite{chou_conjugate_exponential,dytso_poisson_conditional_mean_IT2020,chen1986bounds}. Conversely, different prior families can induce affine conditional medians in the Poisson model~\cite{barnes2025linearityPoisson}, revealing a mismatch between $L^1$ and $L^2$ linearity. Related linearity phenomena have also been studied for heavy-tailed models: \cite{chouity2026linearestimatorsstablevectors} showed that certain jointly symmetric $\alpha$-stable models admit linear conditional-mean estimators.

\subsection{Contributions}

This work makes the following contributions:
\begin{itemize}
    \item We establish a \emph{stability theory} for Bayesian estimation under $L^1$ and $L^2$ losses in the Gaussian noise model.
    \item In the $L^2$ case, we derive explicit quantitative bounds showing that if the conditional expectation is close to an affine function, then the prior must be close to Gaussian in the Lévy metric.
    \item In the $L^1$ case, we develop a functional-analytic proof technique based on Hermite expansions and adjoint operators, showing that approximately linear conditional medians imply proximity to the Gaussian family.
\end{itemize}

The paper outline is as follows. Section~\ref{sec:technical_preliminaries} is dedicated to some technical preliminaries, Section~\ref{sec:main_results} presents our main results, and in Section~\ref{sec:conclusion}, we conclude the paper.

\section{Technical Preliminaries}
\label{sec:technical_preliminaries}
In this section, we review and present some required technical preliminaries. 

\subsection{Monotonicity of 
Bayesian Estimators}

We will need the property that the conditional median estimator is a non-decreasing function, which is shown next.
\begin{lem} \label{lem:median_decreasing}
    Let $X$ be a continuous random variable. Then the conditional median $y \mapsto \mathsf{med}(X|Y=y)$ is a non-decreasing function. 
\end{lem}
\begin{proof}
First note the conditional quantile function $F_{X|Y}$ satisfies
\begin{equation}
p= F_{X|Y}( F_{X|Y}^{-1}(p|y) |y) , \,  p \in (0,1).
\end{equation}
Then, 
\begin{align}
0&=   \frac{\rm d}{ {\rm d} y} F_{X|Y}( F_{X|Y}^{-1}(p|y) |y)= f_{X|Y}( F_{X|Y}^{-1}(p|y) |y)    \frac{\rm d}{ {\rm d} y} F_{X|Y}^{-1}(p|y) + F_{X|Y}'( F_{X|Y}^{-1}(p|y) |y). \label{eq:implicit_expression}
\end{align}

Now, using the meta-derivative identity in \cite{dytso_meta_derivative}, we have that 
\begin{equation}
   \frac{\rm d}{ {\rm d} y} F_{X|Y}( x |y)  =  \frac{\rm d}{ {\rm d} y} \bbE[ 1_{X \le x} |Y=y] = {\rm Cov}( 1_{X \le x}, X |Y=y) .\label{eq:gen_der_of_cdf}
\end{equation} 
Combining \eqref{eq:implicit_expression} and \eqref{eq:gen_der_of_cdf}, we arrive at
\begin{equation}
   \frac{\rm d}{ {\rm d} y} F_{X|Y}^{-1}(p|y)=- \frac{  {\rm Cov}( 1_{ \{ X \le F_{X|Y}^{-1}(p|y) \}},X |Y=y)  }{ f_{X|Y}( F_{X|Y}^{-1}(p|y) |y)   }.
\end{equation}
Now note that $X \mapsto  1_{ \{ X \le t \}}$ is a non-increasing function and $X \mapsto X$ is an increasing function. Therefore, from the FKG inequalities we know ${\rm Cov}( 1_{ \{ X \le F_{X|Y}^{-1}(p|y) \}},X |Y=y)  \le 0$. This implies that the $ \frac{\rm d}{ {\rm d} y} F_{X|Y}^{-1}(p|y) \ge 0$ and concludes the proof. 
\end{proof}

Although we will not need similar properties for the conditional mean, the monotonicity does hold for the conditional mean; see, for example \cite{dytso_meta_derivative}. 

\subsection{Optimality and Orthogonality Conditions}

The conditional mean and conditional median have corresponding optimality conditions which are also known as orthogonality conditions \cite{dytso2017minimum}. These are given next.

\begin{prop} \text{      }  \label{prop:orthog_equations}
    \begin{itemize}
        \item ($L^2$ Case)  For every $g \in L^2(P_Y)$
        \begin{equation}
            \E[(X -\bbE[X|Y]) g(Y) ] =0 .\label{eq:Orthogonality}
        \end{equation}
        \item ($L^1$ Case) For every $g \in L^1(P_Y)$
        \begin{equation}
            \E[\sign(X - \mathsf{med}(X|Y)) g(Y) ] =0 .
            \label{eq:orthog_prop_l1}
        \end{equation}
    \end{itemize}
\end{prop}

To study stability in the $L^2$ case we will need the following proposition. We denote $i =\sqrt{-1}$ and $\phi_X(t) = \bbE[\rme^{itX}]$ as the characteristic function of $X$.

\begin{prop} \label{prop:orthog_fourier} Let $a = \frac{\sigma_X^2}{1+\sigma_X^2}$ where $\sigma_X^2<\infty$ is the variance of $X$. Then for every $g \in L^2(\bbR)$ 
    \begin{equation}
        \bbE[ (X - a Y) g(Y)] = (-i) \int_{-\infty}^\infty \hat{g}(t) \tilde{\phi}_X(t) \rme^{ - \frac{t^2}{2}} \rmd t \; , \label{eq:Orthogonal_expression}
    \end{equation}
    where 
    \begin{align}
        \tilde{\phi}_X(t) &= \frac{1}{1+\sigma^2_X }  \left(  \phi_X'(t)  + \sigma^2_X   t \phi_X( t) \right), \\
        g(x)&= \int_{-\infty}^\infty \hat{g}(t) \rme^{it x} \rmd t \; .
    \end{align} 
\end{prop}
\begin{proof}
Start by noting that
\begin{align}
\E[(X -a Y ) g(Y) ] & =  \E \left[(X -a Y )  \int_\bbR \hat{g}(t) \rme^{it Y} \rmd t \right] \\
&= \int_{-\infty}^\infty  \hat{g}(t)  \E \left[(X -a Y )   \rme^{it Y} \right] \rmd t ,
\end{align}
where the order of integration change follows by Cauchy-Schwarz and Fubini-Tonelli since by assumption $g \in L^2(\bbR)$. 

Now, note that 
\begin{align}
\E[\rme^{itY} (X-aY) ]&=  \E[\rme^{itX} \rme^{it Z} ( ( 1-a)X  -aZ) ] \\
&=( 1-a) \E[X\rme^{itX} \rme^{it Z}]  - a \E[ Z \rme^{itX} \rme^{it Z}  ]\\
&=( 1-a) \E[X\rme^{itX}]   \rme^{-\frac{t^2}{2}}  - a \E[ Z \rme^{it Z}  ] \phi_X( t)\\
&= (1- a) (-i) \phi^{\prime}_X(t)\rme^{-\frac{t^2}{2}}  - a \E[ Z \rme^{it Z}  ] \phi_X( t) \label{eq:Char_Der_Expression}\\
&=( 1-a) (-i) \phi_X'(t) \, \rme^{-\frac{t^2}{2}}  - a (-i)  (-t) \rme^{-\frac{t^2}{2}} \phi_X( t)\\
&=  (-i)   \frac{\rme^{-\frac{t^2}{2}}}{1+\sigma^2_X }  \left(  \phi_X'(t)  + \sigma^2_X   t \phi_X( t) \right) ,
\end{align}
where \eqref{eq:Char_Der_Expression} follow by using that $\phi^{\prime}_X(t)= \E [i X \rme^{itX}]$. 
    \end{proof}
  
\begin{rem}
    Proposition~\ref{prop:orthog_fourier} can be used to establish that Gaussian  is only distribution that induces linearity.  To see this note that from \eqref{eq:Orthogonal_expression}, $ \bbE[ (X - a Y) g(Y)] =0$ for all $g \in L^2(\bbR)$ if and only if 
    \begin{equation}
        \phi_X'(t)  + \sigma^2_X   t \phi_X( t)  =0,  \, \forall t \in \bbR,
    \end{equation}
where the only solution to this differential equation is a Gaussian distribution.
\end{rem}

For the $L^1$ case we will need the following proposition, which is a trivial consequence of the orthogonality-like property in \eqref{eq:orthog_prop_l1}.
\begin{prop} Consider a system of equations
    \begin{align} \label{eq:bd}
T_a[f](y) = \int_{-\infty}^\infty f(x)\mathsf{sign}(x-ay)\exp\left(-\frac{(x-y)^2}{2}\right)dx = 0 \; , \; \forall y\in\mathbb{R}
\end{align}
for some $0\leq a <1$. The density $f$ satisfies\footnote{Here and throughout the paper, if we say \eqref{eq:bd} is satisfied, we mean the equation holds for \emph{almost all} $y$ with respect to Lebesgue measure.} \eqref{eq:bd} if and only if the conditional median
\begin{align}
\mathsf{med}(X|Y=y) = ay \; ,  \forall y\in\mathbb{R}. 
\end{align}
\end{prop}
\noindent We restrict our attention to $a\in [0,1)$ since these are the only admissible values of $a$ \cite{L1EstimationJournal}.

\subsection{Hermite Functions}
For $n \in \{0,1,2,\ldots \}$ let $H_n$ be the Hermite functions defined by
\begin{equation}
H_n(x)=\frac{1}{K_n}e^{\frac{x^2}{2 \sigma^2}} \frac{\rmd^n }{ \rmd x^n} e^{-\frac{x^2}{ \sigma^2}}
\end{equation} 
with $\sigma^2 = \frac{a}{1-a}$. The normalization factors
\begin{align}
K_n = \frac{\pi^\frac{1}{4}2^\frac{n}{2}\sqrt{n!}}{\sigma^{n-\frac{1}{2}}}
\end{align}
are set so that $\|H_n\|_2 = 1$. For $k, n \in  \{0,1,2,\ldots \}$ and $k \neq n$ 
\begin{equation}
\int_{-\infty}^\infty H_k(x) H_n(x)  \rmd x= 0 \; . 
\end{equation} 
In other words, $ \{ H_n(x) \}_{n=0}^\infty$ forms an orthonormal basis. This basis is known to be complete for $L^2(\mathbb{R})$ \cite{hermites}.

\subsection{The Adjoint Operator $T_a^*$} 
In this section we define and study the \emph{adjoint} of $T_a$.
 We view $T_a$ as a mapping from $L^2(\bbR)$ to itself, i.e., $T_a: L^2(\bbR) \to L^2(\bbR)$, with the usual inner product on $L^2(\bbR)$. We use $\mathcal{F}$ to denote the Fourier transform for $L^2$ functions with convention
\begin{align}
\mathcal{F}[f](x) = \int f(t)e^{-2\pi i x t}dt \; .
\end{align}
Recall that the adjoint operator is given by 
\begin{equation}
T_a^*[f](x) =  \int_{-\infty}^\infty k_a(x,y)  f(y) \rmd y =\int_{-\infty}^\infty  \sign(x-ay) \exp\left(-\frac{(y-x)^2}{2}\right) f(y) \rmd y, \, x \in \bbR.  \label{eq:struction_of_adjoint}
\end{equation} 
\begin{prop}
For $a>0$ and $f \in L^2(\bbR)$
\begin{align}
T^*_a[f ](x)&=-T_{\frac{1}{a}}[f ](x), \, x\in \bbR. 
\end{align}

\end{prop}
\begin{proof}
Starting with the structure of the adjoint, we have that
\begin{align}
T^*_a[f ](x)&=\int_{-\infty}^\infty  \sign(x-ay) \exp\left(-\frac{(y-x)^2}{2}\right) f(y) \rmd y\\
&= - \int_{-\infty}^\infty  \sign \left(y- \frac{x}{a} \right) \exp\left(-\frac{(x-y)^2}{2}\right) f(y) \rmd y\\
&=-T_{\frac{1}{a}}[f ](x).
\end{align}
\end{proof} 

\begin{prop}\label{prop:Hermite_output} For $n\geq 1$, there exist functions $\varphi_n$ such that
\begin{equation} 
T^*_a [\varphi_n] = H_n,
\end{equation} 
with the following properties:
\begin{itemize}
\item[i.] $\| e^{-\frac{(1-a)y}{2}}  \varphi_n \|_2 <\infty $
\item[ii.] $\|\mathcal{F}[e^{-\frac{(1-a)y}{2}}\varphi_n]\|_1 = \mathcal{O}(n^\frac{1}{4}) \; .$
\end{itemize}
\end{prop} 
\begin{proof}
Without yet knowing that such functions exist, consider the following formal manipulations:

\begin{align}
H_n(ax) & = \int_{-\infty}^\infty  \sign(ax-ay) \exp\left(-\frac{(ax-y)^2}{2}\right) \varphi_n(y) \rmd y\ \\
& = e^{\frac{a(1-a)}{2}x^2}\int_{-\infty}^\infty\sign(x-y) \exp\left(-\frac{a(x-y)^2}{2}\right) e^{-\frac{(1-a)}{2}y^2} \varphi_n(y) \rmd y
\end{align}
and
\begin{align}
\mathcal{F}\left[e^{-\frac{a(1-a)}{2}x^2}H_n(ax)\right] = \mathcal{F}\left[\sign(x)\exp\left(-\frac{ax^2}{2}\right)\right]\mathcal{F}\left[ e^{-\frac{(1-a)}{2}x^2} \varphi_n(x)\right] \; .
\end{align}
In light of this, we set
\begin{align}
e^{-\frac{(1-a)}{2}x^2} \varphi_n(x) = \mathcal{F}^{-1}\left[\frac{\mathcal{F}\left[e^{-\frac{a(1-a)}{2}x^2}H_n(ax)\right]}{\mathcal{F}\left[\sign(x)\exp\left(-\frac{ax^2}{2}\right)\right]}\right](x) \; . \label{eq:define_varphi}
\end{align}
In order to ensure $\varphi_n$ is well-defined, and that property \emph{i} holds, the claim is that the function inside the large brackets above is square-integrable. The numerator is equal to
\begin{align}
\mathcal{F}\left[e^{-\frac{a(1-a)}{2}x^2}H_n(ax)\right](\omega) = \sqrt{\frac{\pi}{a(1-a)}}\frac{1}{aK_n}\left(\frac{2\pi i}{a}\right)^n \omega^ne^{-\frac{\pi^2\omega^2}{a(1-a)}} \; 
\end{align}
while the denominator is equal to
\begin{align}
c(\omega)= \sqrt{ \frac{1}{a} }    \exp \left(  \frac{- 4 \pi^2 \omega^2}{2a} \right)  \erf \left(  \frac{ i 2\pi w }{  \sqrt{2  a } }\right) \; .
\end{align}  
We note that the function $c(\omega)$ is also known as \emph{the Dawson function} and, in particular, we have the following bounds on this function \cite{janssen2021bounds}:
 \begin{align}
\frac{ \rme^{-\omega^2}  | \rme^{\omega^2}-1|}{2 |\omega| } \le  \left | \sqrt{a} c \left( \frac{\sqrt{a}}{\sqrt{2} \pi} \omega \right) \right|  \le \frac{ \rme^{-\omega^2}  | \rme^{\omega^2}-1|}{|\omega| }. \label{eq:Bounds_On_D_function} 
\end{align}
For $n\geq 1$, both numerator and denominator in \eqref{eq:define_varphi} have simple zeros at the origin, and the bounds \eqref{eq:Bounds_On_D_function} show the ratio is square-integrable.

Property \emph{ii} can be verified as follows:
\begin{align}
\left|\frac{\mathcal{F}\left[e^{-\frac{a(1-a)}{2}x^2}H_n(ax)\right](\omega)}{c(\omega)} \right| & \leq  \frac{C}{K_n}\left(\frac{2\pi i}{a}\right)^n \big|\omega^{n+1} +\omega^{n-1}\big|e^{-\frac{\pi^2\omega^2}{a(1-a)}},
\end{align}
for a constant $C$ that is independent of $n$, and therefore
\begin{align}
\|\mathcal{F}[e^{-\frac{(1-a)y}{2}}\varphi_n]\|_1 \le  \frac{C}{K_n}\left(\frac{2\pi i}{a}\right)^n  \left( I_{n+1} + I_{n-1}\right) ,\label{eq:triangle_in}
\end{align}
where 
\begin{equation}
I_m = \int_{-\infty}^\infty|\omega|^m \rme^{-\beta \omega^2} \rmd \omega  = \beta^{-  \frac{m+1}{2}  } \Gamma \left( \frac{m+1}{2} \right) ,\label{eq:I_m terms}
\end{equation}
and where $\beta = \frac{\pi^2}{ a(1-a)}$.  Next, note that
\begin{equation}
\frac{1}{K_n} \left(\frac{2\pi i}{a}\right)^n \beta^{ -\frac{n}{2}} =\frac{C_a}{\sqrt{n!}} 2^{ \frac{n}{2}}, \label{eq:combined_constans}
\end{equation}
where $C_a$ is a constant that depends on $a$ only.  Now, combining \eqref{eq:triangle_in}, \eqref{eq:I_m terms} and \eqref{eq:combined_constans}, we arrive at
\begin{align}
\left \|\mathcal{F}[e^{-\frac{(1-a)y}{2}}\varphi_n] \right\|_1 \le \frac{C_a'}{\sqrt{n!}} 2^{ \frac{n}{2}} \left( \Gamma \left( \frac{n+2}{2} \right) +\Gamma \left( \frac{n}{2} \right) \right) \le 2  \frac{C_a'}{\sqrt{n!}} 2^{ \frac{n}{2}}  \Gamma \left( \frac{n+2}{2} \right) .
\end{align}
Now using Stirling's approximation,
\begin{align}
\frac{ 2^{ \frac{n}{2}}}{\sqrt{n!}}  \Gamma \left( \frac{n+2}{2} \right)  \approx  \frac{2^{ \frac{n}{2}}}{ (2 \pi)^{1/4} n^{n/2+1/4}} \sqrt{2 \pi}  \left( \frac{n}{2} \right)^{ \frac{n+1}{2}} \propto n^{1/4} .
\end{align}
This concludes the proof. 
\end{proof} 

\section{Main Results}
\label{sec:main_results}
This section is dedicate to our main results. We start with the exposition for the $L^2$ case.
\subsection{Stability in the $L^2$ Case}

We show that if the conditional expectation is close to a linear function in mean-squared error, then the input distribution must be close to a Gaussian in the L\'evy metric, which is defined next.
Let $P$ and $Q$ be two cumulative distribution functions. Then, the \emph{L\'evy distance} between   $P$ and $Q$ is defined as 
\begin{equation}
\mathsf{Levy}(P,Q)  = \inf  \left\{  h  \ge 0:   Q(x-h)-h \le P(x)        \le  Q(x+h)+h,   \forall x \in \mathbb{R} \right\}. 
\end{equation} 
An important property of the L\'evy metric is that  convergence of distributions in the L\'evy metric is equivalent to  weak convergence of distributions \cite{dudley2002real}.  We note that the proof of stability for the $L^2$ case was also considered in \cite{du2017strong}, but we provide a stream-lined version with improved constants for completeness.

\begin{thm}  \label{thm:L2case} Let $a =\frac{\sigma_X^2}{1+\sigma_X^2}$ where $\sigma_X^2$ is the variance of $X$. Suppose that 
\begin{equation}
\E[ |aY  - \E[X|Y]|^2]  \le \epsilon, 
\end{equation} 
then 
\begin{equation}
\mathsf{Levy} \left(P_X , \mathcal{N}(0,\sigma^2_X) \right)   \le  \min_{\delta >0}  \left [\frac{2 (1+\sigma_X^2) }{\pi } \epsilon^{\frac{\delta}{2(1+\delta)}}    +\frac{24}{\pi\sqrt{\frac{2\pi\sigma_X^2}{1+\delta} \log \frac{1}{\epsilon}}} \right] . 
\end{equation} 
\end{thm}
\begin{proof}
    As a first step, note that by using the Cauchy-Schwarz inequality, we have that  
\begin{align}
\sqrt{\E[ |aY  - \E[X|Y]|^2]} &\ge \frac{ | \E[  (  aY  - \E[X|Y]  )  g(Y)  ]   |}{ \sqrt{\E[ |g(Y)|^2 ]}} \\
&= \frac{ | \E[  (    X-aY  )  g(Y)  ]   |}{ \sqrt{\E[ |g(Y)|^2 ]}} \label{eq:using_orthog}\\
&=  \frac{ \left| \int_{\bbR} \hat{g}(t) \tilde{\phi}_X(t) \rme^{ - \frac{t^2}{2}} \rmd t\right|}{   \sqrt{\E[ |g(Y)|^2 ]} } ,\label{eq:f_orth_exp}
\end{align} 
where in \eqref{eq:using_orthog} we have used that 
\begin{align}
\E[  (  aY  - \E[X|Y]  )  g(Y)  ]  &= \E[  (  aY  - X )  g(Y)  ]  + \E[  (  X  - \E[X|Y]  )  g(Y)  ] =\E[  (  aY  - X )  g(Y)  ]  + 0,
\end{align} 
and the orthogonality principle in \eqref{eq:Orthogonality} which states that $\E[  (  X  - \E[X|Y]  )  g(Y)  ]  =0$; and in \eqref{eq:f_orth_exp} we have used Proposition~\ref{prop:orthog_fourier}.

Now choosing 
\begin{equation}
    \hat{g}(t) =  \rme^{\frac{ (\sigma^2_X+1) t^2 } {2}} \rme^{-\frac{\sigma^2_X \tau^2}{2}} 1_{[0, \tau] }(t) ,   \quad t,\tau \in \bbR, 
\end{equation}
in \eqref{eq:f_orth_exp}, we arrive at
\begin{align}
\sqrt{\epsilon}  \cdot  \sqrt{\E[ |g(Y)|^2 ]}  & \ge  \frac{1}{1+\sigma_X^2} \rme^{-\frac{\sigma^2_X \tau^2}{2}}   \left|   \int_0^{\tau}  \rme^{\frac{\sigma^2_X t^2 }{2}} \left(  \phi_{X}^{\prime}(t )   + \sigma^2_X t   \phi_{X}(t )    \right) \rmd t \right|\\
&= \frac{1}{1+\sigma_X^2}  \left| \phi_{X}(\tau) -\rme^{-\frac{\sigma^2_X \tau^2}{2}} \right| \label{eq:bound_diff_char}
\end{align}
where \eqref{eq:bound_diff_char} follows from Lemma~\ref{lem:diff_char} (below).  

Additionally note that 
\begin{align}
   \sqrt{ \E[ |g(Y)|^2 ] }&= \sqrt{\E \left[ | \int_\bbR \hat{g}(t) \rme^{it Y} \rmd t|^2 \right] }\\
    & \le  |\tau|  \rme^{\frac{ (\sigma^2_X+1) \tau^2 } {2}} \rme^{-\frac{\sigma^2_X \tau^2}{2}}  = |\tau|  \rme^{\frac{ \tau^2}{2}} \label{eq:bound_on_moments_g}.
\end{align}



Now that we have established proximity bounds for the characteristic function, we can establish proximity bounds for the CDFs. To that end, recall Esseen's inequality \cite{feller1971introduction}: Let $X$ and $Y$ be two random variables with CDFs $F_X$ and $F_Y$, respectively, and assume that $Y$ has a PDF, then for any $T>0$
\begin{equation}
\sup_{t \in \bbR} | F_X(t) - F_Y(t) | \le \frac{1}{\pi} \int_{-T}^T  \left| \frac{\phi_X(\omega) -\phi_Y(\omega)}{\omega} \right|  \rmd \omega  + \frac{ 24 \| f_Y\|_\infty }{\pi T} .\label{eq:Esseen_Ineq}\end{equation} 

Therefore, combining \eqref{eq:bound_diff_char}, \eqref{eq:bound_on_moments_g} and \eqref{eq:Esseen_Ineq}, we arrive at 
\begin{align}
\sup_{t \in \bbR} | F_X(t) - \Phi_{\sigma_X}(t) | & \le  \frac{ \sqrt{\epsilon} (1+\sigma_X^2) }{\pi } \int_{-T}^T \rme^{ \frac{\omega ^2}{2}}  \rmd \omega  +\frac{24}{\pi T \sqrt{2\pi\sigma_X^2}}\\
 & \leq \frac{\sqrt{\epsilon} (1+\sigma_X^2) }{\pi } 2 \rme^{ \frac{T ^2}{2}}   +\frac{24}{\pi T \sqrt{2\pi\sigma_X^2}} \; .\label{eq:dawson_step}
\end{align}
 In the step \eqref{eq:dawson_step} above, we use the boundedness of the Dawson integral. Choosing $T =\sqrt{\frac{1}{1+\delta} \log \frac{1}{\epsilon}}$ gives
 \begin{align}
\sup_{t \in \bbR} | F_X(t) - \Phi_{\sigma_X}(t) | & \le  \frac{2 (1+\sigma_X^2) }{\pi } \epsilon^{\frac{\delta}{2(1+\delta)}}    +\frac{24}{\pi\sqrt{\frac{2\pi\sigma_X^2}{1+\delta} \log \frac{1}{\epsilon}}}
\end{align}
as desired. Minimizing over $\delta$ concludes the proof. 
\end{proof}

\subsection{Stability in the $L^1$ Case}

Our main result on stability for the optimal $L^1$ estimator is the following. Note that by taking $\varepsilon=0$, the stability result recovers our previous uniqueness result from \cite{L1EstimationJournal} (except assuming that $X$ has a bounded square integrable density $f$), but does so in a purely functional manner, relying only on the $L^2$ Hermite expansion and the properties of the adjoint operator $T^*_a$ instead of any results from the theory of tempered distributions.
\begin{thm} \label{thm:L1case}
Let $\psi(y) = \mathsf{med}(X|Y=y)$  and assume that the density of $X$ satisfies $f(x)\leq M$ for all $x$. For any $\varepsilon\geq 0$, if
\begin{align} \label{assumption1}
    \int_{-\infty}^\infty \exp\left(\frac{1-a}{2}y^2-\min_{x\in\left[ay-\sqrt{2a\varepsilon},ay+\sqrt{2a\varepsilon}\right]}\frac{(x-y)^2}{2}\right)\big|ay - \psi(y)\big|\rmd y \leq B < \infty 
\end{align}
and
\begin{align} \label{eq:assumption2}
\left|\psi^{-1}(x) - \frac{x}{a}\right| \leq \Delta(x) \; , \quad \sum_i b_i \leq \varepsilon \; ,
\end{align}
where $\Delta(x) = b_i$ when $x\in[i,i+1)$,
then
\begin{equation} |\langle f,H_n \rangle| \leq L_n\varepsilon,
\end{equation}
for $n\geq 1$ and $L_n = \mathcal{O}(n^\frac{1}{4})$.
\end{thm}
Before getting to the proof, let us consider the assumptions of the theorem. The condition \eqref{assumption1} assumes that $\psi(y)$ and $ay$ differ only by a function that decays like a Gaussian. In particular, if $\varepsilon$ is small, the difference needs to decay like $\exp\left(-\frac{a(1-a)}{2}y^2\right)$. Fortunately, it can be checked that this assumption is indeed satisfied by any density $f(x)$ that is given by $\mathcal{N}\left(0,\frac{a}{1-a}\right)+\xi(x)$ where $\xi$ is a compactly supported perturbation, and the set of such densities is dense in the ambient space of all square-integrable $f(x)$. The condition \eqref{eq:assumption2} is a slightly stronger and more uniform version of $\psi(y)$ and $ay$ being close in $L^1$ distance. Note that
\begin{align}
\int_{-\infty}^\infty |\psi(y) - ay|\rmd y & = \int_{-\infty}^\infty \big|\psi^{-1}(x) - \frac{x}{a}\big| \rmd x \\
& \leq \int_{-\infty}^\infty \Delta(x)\rmd x = \sum_i b_i \leq \varepsilon \; .
\end{align}
The punchline of the theorem, that $|\langle f,H_n \rangle|$ is small for $n \geq 1$, means that $f$ must be approximately orthogonal to every Hermite function in the complete orthonormal sequence other than the Gaussian $H_0$, and in this sense $f$ must be close to a Gaussian. In the corollary that follows, we will show how this can imply $f$ is close to the density $\mathcal{N}\left(0,\frac{a}{1-a}\right)$ in $L^2$ distance provided that the Hermite coefficients of $f$ decay fast enough.

The assumptions in Theorem \ref{thm:L1case} are much stronger than those in Theorem \ref{thm:L2case}. In particular, we assume the optimal estimator is pointwise very close to a linear function, instead of just close on average. However, the resulting stability is also stronger (even though it is in the style of ``weak'' convergence) in the sense that it is functional stability of the density $f$, instead of a statement about CDFs. Note that the techniques used in Theorem \ref{thm:L2case} have no hope of working unmodified in the $L^1$ case, since they do not, in a fundamental way, use the non-negativity of $f$, and stability fails when $f$ can be negative. This necessitated the new assumptions and techniques in the following proof.

\begin{proof}
Consider the orthonormal Hermite function expansion
\begin{equation}
f(x) = \sum_{n=0}^\infty c_iH_i(x) \; .
\end{equation}
Our goal is to show that under the assumptions of the theorem, all $c_i$ coefficients must be small except for $c_0$ which corresponds to the Gaussian $H_0(x)$.

Defining $\widetilde{f}(x) = e^{\frac{x^2(1-a)}{2a}}f(x)$, we will first show that
\begin{align}
\int_{z-0.5}^{z+0.5}\widetilde{f}(ax)dx \leq C_0
\end{align}
for any $z$ and some constant $C_0$. To this end, note that according to Proposition~\ref{prop:orthog_equations} the conditional median $\psi(y)$ satisfies 
\begin{equation} \label{eq:cond_med_eq}
\int_{-\infty}^\infty f(x)\mathsf{sign}(x-\psi(y))\exp\left(-\frac{(x-y)^2}{2}\right) \rmd x = 0
\end{equation}
by definition. Then using the above display, \eqref{eq:bd} can be rewritten as
\begin{align}
T_a[f](y) & = \int_{-\infty}^\infty f(x)\mathsf{sign}(x-ay)\exp\left(-\frac{(x-y)^2}{2}\right) \rmd x \\
& = \int_{-\infty}^\infty f(x)\big[\mathsf{sign}(x-ay)-\mathsf{sign}(x-\psi(y))\big]\exp\left(-\frac{(x-y)^2}{2}\right) \rmd x \; \label{eq:bd_op_plus}
\end{align}
Under the assumptions of the theorem,
\begin{equation}
    \int_{-\infty}^\infty |ay-\psi(y)| \rmd y \leq \varepsilon \; ,
\end{equation}
and since $\psi$ is nondecreasing, as shown in Lemma~\ref{lem:median_decreasing}, this implies
\begin{equation}
     \label{eq:uniform_diff_bound}
    \sup_y |ay-\psi(y)| \leq \sqrt{2a\varepsilon} \; .
\end{equation}
Using \eqref{eq:bd_op_plus}, \eqref{eq:uniform_diff_bound}, and $f(x)\leq M$, 
\begin{align} \label{eq:op_bound}
\big|T_a[f](y)\big| & \leq 2M\big|ay - \psi(y)\big|\exp\left(-\min_{x\in\left[ay-\sqrt{2a\varepsilon},ay+\sqrt{2a\varepsilon}\right]}\frac{(x-y)^2}{2}\right) \; .
\end{align}

The integral equation in \eqref{eq:bd} can be massaged into a convolution operator as was done in \cite[Prop.~5]{L1EstimationJournal},
\begin{align} 
e^\frac{(1-a)y^2}{2}T_a[f](y) & = \int_{-\infty}^\infty \widetilde{f}(ax)\mathsf{sign}(x-y)\exp\left(-\frac{a(x-y)^2}{2}\right) \rmd x \label{eq:conv_1}\\
& = \int_{-\infty}^\infty \widetilde{f}(ax)g'(y-x) \rmd x \label{eq:conv}
\end{align}
for a kernel $g'$ with an antiderivative $g$. Following \eqref{eq:conv} and using dominated convergence to pull the derivative outside of the integral (and then integrating),
\begin{equation}
\int_{-\infty}^\infty \widetilde{f}(ax)g(y-x)dx \leq C_1 + \int_{-\infty}^\infty e^\frac{(1-a)y^2}{2}\big|T_a[f](y)\big| \rmd y \; .
\end{equation}
Then using \eqref{eq:op_bound} along with the assumption \eqref{assumption1},
\begin{equation}
\int_{-\infty}^\infty \widetilde{f}(ax)g(y-x)dx \leq C_2 \; .
\end{equation}
Using the trick from \cite{L1EstimationJournal} that
\begin{equation}
g(x) \geq c1_{[-0.5,0.5]}(x)
\end{equation}
where $c>0$ and $1_{[-0.5,0.5]}$ is the indicator function for the interval $[-0.5,0.5]$, we have
\begin{equation} \label{eq:growth_estimate}
\int_{y-0.5}^{y+0.5} \widetilde{f}(ax) \rmd x \leq C_0
\end{equation}
for all $y$ as desired.

Next, we use the growth estimate from \eqref{eq:growth_estimate} to show that
\begin{equation}
\int_{-\infty}^\infty \big|e^{\frac{(1-a)y^2}{2}}T_af(y)\big| \rmd y \leq C_3\varepsilon \; .
\end{equation}
For this, following \eqref{eq:conv_1} and using \eqref{eq:cond_med_eq},
\begin{align}
e^\frac{(1-a)y^2}{2}T_a[f](y) & = \int_{-\infty}^\infty \widetilde{f}(ax)\left[\mathsf{sign}(ax-ay)-\mathsf{sign}(ax-\psi(y))\right]\exp\left(-\frac{a(x-y)^2}{2}\right) \rmd x
\end{align}
so that
\begin{align}
\big|e^\frac{(1-a)y^2}{2}T_a[f](y)\big| \leq 2\left|\int_{ay}^{\psi(y)}\widetilde{f}(ax) \rmd x\right|
\end{align}
and
\begin{align}
\int_{-\infty}^\infty \big|e^\frac{(1-a)y^2}{2}T_a[f](y)\big| \rmd y \leq 2\int_{-\infty}^\infty \left|\int_{ay}^{\psi(y)}\widetilde{f}(ax) \rmd x\right| \rmd y \; . \label{eq:L1_int}
\end{align}
We can rewrite \eqref{eq:L1_int} as
\begin{align}
\int_{-\infty}^\infty \big|e^\frac{(1-a)y^2}{2}T_a[f](y)\big|\rmd y & \leq 2  \int_{-\infty}^\infty  \left|\int_{\frac{x}{a}}^{\psi^{-1}(x)} \rmd y\right|\widetilde{f}(ax) \rmd x \\
& = 2\int_{-\infty}^\infty \left|\psi^{-1}(x) - \frac{x}{a}\right|\widetilde{f}(ax)\rmd x \\
& \leq 2\int_{-\infty}^\infty \Delta(x)\widetilde{f}(ax) \rmd x \label{eq:penultimate_L1_bound}\\
& \leq 2 C_0 \sum_{i} b_i \label{eq:final_L1_bound} \; .
\end{align}
The inequality \eqref{eq:penultimate_L1_bound} uses the assumption \eqref{eq:assumption2}, and the final inequality \eqref{eq:final_L1_bound} uses the growth estimate \eqref{eq:growth_estimate}. We have therefore shown that the assumptions
\begin{equation}
\left|\psi^{-1}(x) - \frac{x}{a}\right| \leq \Delta(x) \; , \quad \sum_i b_i \leq \varepsilon \; ,
\end{equation}
where $\Delta(x) = b_i$ when $x\in[i,i+1)$, imply
\begin{equation} \label{eq:step2}
\int_{-\infty}^\infty \big|e^\frac{(1-a)y^2}{2}T_a[f](y)\big| \rmd y \leq C_3\varepsilon \; .
\end{equation}

In Proposition \ref{prop:Hermite_output} in Appendix~\ref{sec:Aux_res}, we show how to construct functions $\varphi_n$ for $n\geq 1$ such that $\varphi_n$ is mapped by the adjoint of $T_a$ to $H_n$. Letting $\phi_0(y) = e^{-\frac{(1-a)y^2}{2}}$, these $\varphi_n$ have the nice property that $\|\phi_0\varphi_n\|_2$ are finite. We have:
\begin{align}
\big| \langle T_af, \varphi_n \rangle \big| = \big|\langle \phi_0^{-1}T_af,\phi_0\varphi_n \rangle \big| \leq C_3\varepsilon\|\phi_0\varphi_n\|_\infty \leq C_3\varepsilon\|\mathcal{F}[\phi_0\varphi_n]\|_1
\end{align}
for $n\geq 1$. Furthermore, Proposition \ref{prop:Hermite_output} also shows that
\begin{equation}
\|\mathcal{F}[\phi_0\varphi_n]\|_1 =\mathcal{O}(n^\frac{1}{4})
\end{equation}
for some $d<1$. Using Proposition \ref{prop:absolute_integrability} below to justify the interchange of limits,
\begin{equation}
\langle T_af,\varphi_n\rangle = \langle f, T_a^*\varphi_n\rangle = \langle f, H_n\rangle = c_n \; .
\end{equation}
Putting this all together,
\begin{equation}
|c_n| = |\langle f,H_n \rangle| \leq L_n \varepsilon
\end{equation}
for $n\geq 1$ where $L_n = \mathcal{O}(n^\frac{1}{4})$.

\end{proof}

\begin{cor}
Under the assumptions of Theorem \ref{thm:L1case} and $f\in L^2(\mathbb{R})$, suppose that the Hermite coefficients $c_n$ additionally satisfy
\begin{equation}
\sum_{n=0}^\infty |c_n|n^\frac{1}{4} < \infty \; .
\end{equation}
Then,
\begin{equation}
\left\|f-c_0H_0\right\|_2^2 \leq C\varepsilon \; .
\end{equation}
\end{cor}
\begin{proof}
We have
\begin{equation}
\left\|f(x) - c_0H_0\right\|_2^2 = \sum_{n=1}^\infty c_n^2 \leq \sum_{n=1}^\infty |c_n|L_n\varepsilon \leq C\varepsilon \; .
\end{equation}
\end{proof}
Fast decay of Hermite coefficients occurs for many functions of interest, and in particular it is known that the coefficients $c_n$ will decay exponentially if and only if both $f$ and its Fourier transform have Gaussian decay (with any variance) \cite{hermites}.

\section{Conclusion} \label{sec:conclusion}
This work establishes a functional framework that links the approximate linearity of Bayesian estimators to the near-Gaussianity of the underlying prior. For the $L^2$ loss, we obtained an explicit quantitative rate in the L\'evy metric, while for the $L^1$ loss we introduced a Hermite--analytic method establishing that the Gaussian the unique stable solution.   Several refinements remain open, particularly for the $L^1$ case such as relaxing the technical decay assumptions. For the $L^2$ case it would also be interesting to find if the current order in terms of $\epsilon$ is tight.

\bibliography{refs.bib}

@article{hermites,
  title={Hermite expansions for spaces of functions with nearly optimal time-frequency decay},
  author={Neyt, Lenny and Toft, Joachim and Vindas, Jasson},
  journal={Journal of Functional Analysis},
  volume={288},
  number={3},
  pages={110706},
  year={2025},
  publisher={Elsevier}
}

@article{dytso2017minimum,
  title={On the minimum mean $ p $ th error in {G}aussian noise channels and its applications},
  author={Dytso, Alex and Bustin, Ronit and Tuninetti, Daniela and Devroye, Natasha and Poor, H Vincent and Shitz, Shlomo Shamai},
  journal={IEEE Transactions on Information Theory},
  volume={64},
  number={3},
  pages={2012--2037},
  year={2017},
  publisher={IEEE}
}

@article{akyol2012conditions,
  title={On conditions for linearity of optimal estimation},
  author={Akyol, Emrah and Viswanatha, Kumar and Rose, Kenneth},
  journal={IEEE Transactions on Information Theory},
  volume={58},
  number={6},
  pages={3497--3508},
  year={2012},
  publisher={IEEE}
}

@misc{chouity2026linearestimatorsstablevectors,
      title={On Linear Estimators for some Stable Vectors}, 
      author={Rayan Chouity and Charbel Hannoun and Jihad Fahs and Ibrahim Abou-Faycal},
      year={2026},
      eprint={2601.09554},
      archivePrefix={arXiv},
      primaryClass={cs.IT},
      url={https://arxiv.org/abs/2601.09554}, 
}

@article{diaconis1979conjugate,
  title={Conjugate priors for exponential families},
  author={Diaconis, Persi and Ylvisaker, Donald},
  journal={The Annals of Statistics},
  volume={7},
  number={2},
  pages={269--281},
  year={1979},
  publisher={JSTOR}
}

@book{dudley2002real,
  author    = {Richard M. Dudley},
  title     = {Real Analysis and Probability},
  series    = {Cambridge Studies in Advanced Mathematics},
  volume    = {74},
  publisher = {Cambridge University Press},
  year      = {2002},
  edition   = {2nd},
  address   = {Cambridge}
}

@book{feller1971introduction,
  author    = {William Feller},
  title     = {An Introduction to Probability Theory and Its Applications, Volume II},
  publisher = {John Wiley \& Sons},
  year      = {1971},
  edition   = {2nd},
  address   = {New York}
  
}

@article{janssen2021bounds,
  title={Bounds on {D}awson’s integral occurring in the analysis of a line distribution network for electric vehicles},
  author={Janssen, AJEM},
  journal={EURANDOM Preprint Series},
  volume={14},
  year={2021}
}

@article{chen1986bounds,
  title={Bounds for the difference between median and mean of gamma and {P}oisson distributions},
  author={Chen, Jeesen and Rubin, Herman},
  journal={Statistics \& Probability letters},
  volume={4},
  number={6},
  pages={281--283},
  year={1986},
  publisher={North-Holland}
}

@ARTICLE{L1EstimationJournal,
  author={Barnes, Leighton P. and Dytso, Alex and Liu, Jingbo and Vincent Poor, H.},
  journal={IEEE Transactions on Information Theory}, 
  title={L1 Estimation: On the Optimality of Linear Estimators}, 
  year={2024},
  volume={70},
  number={11},
  pages={8026-8039},
  doi={10.1109/TIT.2024.3440929}}

@book{kailath2000linear,
  title={Linear Estimation},
  author={Kailath, Thomas and Sayed, Ali H and Hassibi, Babak},
  year={2000},
  publisher={Prentice Hall}
}

@ARTICLE{dytso_meta_derivative,
  author={Dytso, Alex and Poor, H. Vincent and Shamai Shitz, Shlomo},
  journal={IEEE Transactions on Information Theory}, 
  title={Conditional Mean Estimation in Gaussian Noise: A Meta Derivative Identity With Applications}, 
  year={2023},
  volume={69},
  number={3},
  pages={1883-1898},
  doi={10.1109/TIT.2022.3216012}}

@inproceedings{barnes_multivariate_gaussian_ISIT2024,
  author    = {Barnes, Leighton P. and Dytso, Alex and Liu, Jingbo and Poor, H. Vincent},
  title     = {Multivariate priors and the linearity of optimal Bayesian estimators under Gaussian noise},
  booktitle = {Proc. IEEE International Symposium on Information Theory (ISIT)},
  year      = {2024}
}

@misc{barnes2025linearityPoisson,
      title={Linearity-Inducing Priors for Poisson Parameter Estimation Under $L^{1}$ Loss}, 
      author={Leighton P. Barnes and Alex Dytso and H. Vincent Poor},
      year={2025},
      eprint={2505.21102},
      archivePrefix={arXiv},
      primaryClass={math.ST},
      url={https://arxiv.org/abs/2505.21102}, 
}

@article{dytso_poisson_conditional_mean_IT2020,
  author  = {Dytso, Alex and Poor, H. Vincent and Shamai, Shlomo},
  title   = {Optimal Estimation for the Poisson Noise Model: A Characterization of Linear Conditional Means},
  journal = {IEEE Transactions on Information Theory},
  year    = {2020},
  volume  = {66},
  number  = {10},
  pages   = {6306--6329}
}

@article{chou_conjugate_exponential,
  author  = {Chou, Kai-Tai and Olkin, Ingram},
  title   = {Completely Monotone Functions and Characterizations of the Gamma Distribution},
  journal = {Annals of the Institute of Statistical Mathematics},
  volume  = {21},
  pages   = {601--617},
  year    = {1969}
}

@article{du2017strong,
  title={Strong data processing inequalities for input constrained additive noise channels},
  author={du Pin Calmon, Flavio and Polyanskiy, Yury and Wu, Yihong},
  journal={IEEE Transactions on Information Theory},
  volume={64},
  number={3},
  pages={1879--1892},
  year={2017},
  publisher={IEEE}
}
\bibliographystyle{IEEEtran}

\begin{appendices}
\section{Auxiliary Results}
\label{sec:Aux_res}

\begin{lem}\label{lem:diff_char} For $\tau \in \bbR$
\begin{align*}
\left| \phi_{X}(\tau) -\rme^{-\frac{\sigma^2_X \tau^2}{2}} \right|
= \rme^{-\frac{\sigma^2_X \tau^2}{2}}   \left|   \int_0^{\tau}  \rme^{\frac{\sigma^2_X t^2 }{2}} \left(  \phi_{X}^{\prime}(t )   + \sigma^2_X t   \phi_{X}(t )    \right) \rmd t \right|
\end{align*} 
\end{lem}

\begin{proof}
Next  we use  the following  observations
\begin{align}
\phi_{X}(\tau)  \rme^{\frac{\sigma^2_X \tau^2}{2}} -1 &= \int_0^\tau   \phi_{X}^{\prime}(u)   \rme^{\frac{\sigma^2_X u^2}{2}} +  \phi_{X}(u)  \rme^{\frac{\sigma^2_X u^2}{2} } \sigma^2_X u  \, \rmd u \label{eq:usingFTC} \\
&=\int_0^\tau   \rme^{\frac{\sigma^2_X u^2}{2}} \left(  \phi_{X}^{\prime}(u)   +  \phi_{X}(u)   \sigma^2_X u \right) \rmd u, \label{eq:ExpressionAfterFTC}
\end{align}
where the expression in \eqref{eq:usingFTC} observing that $\frac{\rmd }{ \rmd \tau} \phi_{X}(\tau)  \rme^{\frac{\sigma^2_X \tau^2}{2} }=   \phi_{X}^{\prime}(\tau)  \rme^{\frac{\sigma^2_X \tau^2}{2}}+ \phi_{X}(\tau)  \rme^{\frac{\sigma^2_X \tau^2}{2}} \sigma^2_X \tau$ and using fundamental theorem of calculus  $f(b)-f(a)=\int_a^b f^{\prime}(u)  \rmd u$.

Next, by using  \eqref{eq:ExpressionAfterFTC}, we have
\begin{align*}
\left| \phi_{X}(\tau) -\rme^{-\frac{\sigma^2_X \tau^2}{2}} \right|
= \rme^{-\frac{\sigma^2_X \tau^2}{2}}   \left|   \int_0^{\tau}  \rme^{\frac{\sigma^2_X t^2 }{2}} \left(  \phi_{X}^{\prime}(t )   + \sigma^2_X t   \phi_{X}(t )    \right) \rmd t \right| . 
\end{align*} 
\end{proof} 

\begin{prop} \label{prop:absolute_integrability}
For $\varphi_n$ defined as in Proposition \ref{prop:Hermite_output}, and under the assumptions of Theorem \ref{thm:L1case},
\begin{equation}
\langle T_af,\varphi_n\rangle = \langle f, T_a^*\varphi_n\rangle
\end{equation}
\end{prop}
\begin{proof}
We use Fubini's theorem to switch the order of integration,
\begin{align}
\int_{-\infty}^\infty T_af(y)\varphi_n(y)dy & = \int_{-\infty}^\infty \int_{-\infty}^\infty f(x)\mathsf{sign}(x-ay)\exp\left(-\frac{(x-y)^2}{2}\right)\varphi_n(y) \rmd x \rmd y \\
& =  \int_{-\infty}^\infty f(x)\left[ \int_{-\infty}^\infty \mathsf{sign}(x-ay)\exp\left(-\frac{(x-y)^2}{2}\right)\varphi_n(y) \rmd y\right] \rmd x \\
& = \int_{-\infty}^\infty f(x) T_a^*\varphi_n(x) \rmd x \; .
\end{align}
Absolute integrability is verified as follows:
\begin{align}
\int_{-\infty}^\infty \int_{-\infty}^\infty & \bigg| \widetilde{f}(x)\mathsf{sign}(x-ay)\exp\left(-\frac{\left(x - ay\right)^2}{2a}\right)\varphi_n(y)\exp\left(-\frac{(1-a)y^2}{2}\right)\bigg| \rmd x \rmd y  \\
& = \int_{-\infty}^\infty |\varphi_n(y)|\exp\left(-\frac{(1-a)y^2}{2}\right)\sum_{z\in\mathbb{Z}}\int_z^{z+1} \widetilde{f}(x)\exp\left(-\frac{\left(x - ay\right)^2}{2a}\right)\rmd x \rmd y \\
& \leq \int_{-\infty}^\infty |\varphi_n(y)|\exp\left(-\frac{(1-a)y^2}{2}\right)\sum_{z\in\mathbb{Z}}C_4\sup_{x\in[z,z+1)}\exp\left(-\frac{\left(x - ay\right)^2}{2a}\right) \rmd y \label{eq:proofstep2}\\
& \leq C_5  \int_{-\infty}^\infty |\varphi_n(y)|\exp\left(-\frac{(1-a)y^2}{2}\right) \rmd y < \infty \; .
\end{align}
In \eqref{eq:proofstep2} we have used the growth estimate from \eqref{eq:growth_estimate}.
\end{proof}
\end{appendices}

\end{document}